\documentclass[conference]{IEEEtran}

\usepackage[dvips]{graphicx}
\usepackage{amsmath,amssymb}
\usepackage{algorithm}
\usepackage{algorithmic}

\usepackage{amsthm}
\usepackage{latexsym}
\usepackage{cite}

\newtheorem{theorem}{Theorem}
\newtheorem{lemma}[theorem]{Lemma}
\newtheorem{corollary}[theorem]{Corollary}

\theoremstyle{definition}
\newtheorem{definition}[theorem]{Definition}
\newtheorem{remark}[theorem]{Remark}

\begin{document}
\title{Weight Spectrum of Quasi-Perfect Binary Codes with Distance 4}
\date{}
\author{
\IEEEauthorblockN{Valentin B. Afanassiev}
\IEEEauthorblockA{Institute for Information Transmission Problems\\
(Kharkevich institute), Russian Academy of
 Sciences\\ GSP-4, Moscow, 127994, Russian Federation\\afanv@iitp.ru\\}
\and
\IEEEauthorblockN{Alexander A. Davydov}
\IEEEauthorblockA{Institute for Information Transmission Problems\\
(Kharkevich institute), Russian Academy of
 Sciences\\ GSP-4, Moscow, 127994, Russian Federation\\adav@iitp.ru\\}
}
\maketitle

\begin{abstract}
We consider the weight spectrum of a class of quasi-perfect binary linear codes with code distance 4. For example, extended  Hamming code and Panchenko code are the known members of this class. Also, it is known that in many cases Panchenko code has the minimal number of  weight 4 codewords. We give exact recursive formulas for the weight spectrum of quasi-perfect codes and their dual codes. As an example of application of the weight spectrum we derive a lower estimate for the conditional probability of correction of erasure patterns of high weights (equal to or greater than code distance).
\end{abstract}

\section{Introduction}
Calculation or estimation of the weight spectrum of linear code is one of very old unresolved problem that gives rise a long list of other unresolved problems in coding theory. Binary quasi-perfect codes has a long history in investigation but with a ``hole'' in area of weight distribution for the most of  the codes. We caught a happy chance to find a ``simple'' solution for weight spectrum of a whole class of binary quasi-perfect codes.

The other and real motivation for the research was to search most effective encoding and decoding schemes for error correction and error detection in computer memory. The physical  volume of contemporary memory cells tends to ``zero'' but the probability of error or defect in a cell tends to be very critical for a whole memory device. As a consequence of this trend we need more and more effective encoding schemes for correction of independent errors and their collections in the form of two dimensional blots.

The binary quasi-perfect extended Hamming code is traditional choice for memory devices. We suggest as a better choice Panchenko code  in original and product forms (for blot correction).
 The main our improvement over the  traditional solution is the extension of the decoding area due to  correction of detected errors as erasure patterns of  weights equal to or greater than the code distance.

\section{Quasi-perfect codes created by the doubling construction}\label{sec_doubconstr}
Let an $[n,n-r,d]$ be a linear binary code of length $n$, redundancy $r$, and minimum distance $d$. For a code with redundancy $r$ we introduce also the following notations: $n_r$ is length of the code, $H_r$ is its  parity check  matrix of size $r\times n_r$, and $d_r$ is code distance.
\begin{definition}
The \emph{doubling construction} creates a parity check  matrix $H_r$ of an $[n_r,n_r-r,d_r]$ code from a parity check  matrix $H_{r-1}$ of an $[n_{r-1},n_{r-1}-(r-1),d_{r-1}]$ code as follows
\begin{equation} \label{eq2_doubl}
H_r =\left[ {\begin{array}{ccc}
0\ldots0&|&1\ldots1\\
- - - -&|&- - -\\
H_{r-1}&|&H_{r-1}
\end{array}} \right].
\end{equation}
\end{definition}

 By \eqref{eq2_doubl} we have $n_{r}=2n_{r-1}$.

  Let us define matrices $S$ and $M$ as
\begin{equation*}
S=\left[\begin{array}{ccccc} {1} & {0} & {0} & {0} & {1} \\ {0} & {1} & {0} & {0} & {1} \\ {0} & {0} & {1} & {0} & {1} \\ {0} & {0} & {0} & {1} & {1} \end{array}\right],\qquad M=\left[\begin{array}{cc}0&1\\1&1\end{array}\right].
\end{equation*}

Denote by $H^{\text{EH}}_r$ a parity check matrix of the extended Hamming $[2^{r-1},2^{r-1}-r,4]$ code, $r\ge3$.
By \eqref{eq2_doubl}, if $H_{r-1}=M$ (resp. $H_{r-1}=H^{\text{EH}}_{r-1}$) then $H_r=H^{\text{EH}}_3$ (resp. $H_r=H^{\text{EH}}_r$).
If in \eqref{eq2_doubl} we have $d_{r-1}=3$ then $d_r=3$ since the left part of $H_r$ contains 3 linear dependent columns provided by
 the structure of  $H_{r-1}$. Finally, let $h_i$ (resp. $[0h_i]^T$ or $[1h_i]^T$) be a column of $H_{r-1}$ (resp. $H_{r}$). If in \eqref{eq2_doubl}
 $d_{r-1}\ge4$ then $d_r=4$ as the sum of columns $[0h_i]^T+[0h_j]^T+[1h_i]^T+[1h_j]^T$, $i\neq j$, is equal to zero.

\begin{definition} A code correcting $t$ errors is \emph{quasi-perfect} if its covering radius is equal to $t+1$.
\end{definition}

In particular, a quasi-perfect code with distance $d=4$ has covering radius 2. Minimum distance of any code correcting $t$ errors is equal to $2t+1$ or $2t+2$.
A linear quasi-perfect code is \emph{``non-extendable"} in the sense that addition of any column to a parity check  matrix decreases the code distance. Any linear $[n,n-r,2t+2]$ code with $2t+2\ge4$ is either a quasi-perfect one or shortening of some quasi-perfect code of redundancy $r$ and distance $2t+2$.

\begin{theorem}\label{th_DT}
\emph{\cite{DavydovTombakPPI}} Let $n_r\ge2^{r-2}+2$, $r\ge3$, and let an $[n_r,n_r-r,4]$ code  be quasi-perfect. Then a parity check  matrix $H_{r}$ of the code can be presented in the form \eqref{eq2_doubl} where matrix $H_{r-1}$ is given in one of the following three variants only:\\
$\bullet$ $H_{r-1}$ is a parity check  matrix of an $[n_{r-1},n_{r-1}-(r-1),4]$ quasi-perfect code with $n_{r-1}=\frac{1}{2}n_{r}$; \\
$\bullet$ $H_{r-1}=S$;\\
$\bullet$ $H_{r-1}=M.$
\end{theorem}
\begin{corollary}\label{cor1_DT}
\emph{\cite{DavydovTombakPPI}} Let $n_r\ge2^{r-2}+2$, $r\ge5$, and let an $[n_r,n_r-r,4]$ code  be quasi-perfect. Then length $n_r$ can take any value from the sequence
\begin{equation}\label{eq2_length}
n_r=2^{r-2}+2^{r-2-g} \text{ for } g=0,2,3,4,5,\ldots,r-3.
\end{equation}
Moreover, for each $g=0,2,3,4,5,\ldots,r-3$, there exists an $[n_r,n_r-r,4]$ quasi-perfect code with $n_r=2^{r-2}+2^{r-2-g}$.  Also, $n_r$ may not take any other value that is not noted in \eqref{eq2_length}.
\end{corollary}

Now we give a general description of a parity check matrix for whole class of quasi-perfect codes with distance $4$. Let
$$
B_{k,g} =\left[b_{k} \ldots b_{k}\right],\quad g\in\{0,2,3,4,5,\ldots,r-3\},
$$
be the  $(r-g-2)\times (2^g+1)$  matrix of identical columns $b_{k}$, where $r\ge5$ is code redundancy, $b_{k}$ is the binary representation of the integer $k$ (with the most significant bit  at the top position).
\begin{corollary}\label{cor2_DT}
\emph{\cite{DavydovTombakPPI}} Let $n_r=2^{r-2}+2^{r-2-g}$, $r\ge5$, $g\in\{0,2,3,4,5,\ldots,r-3\}$, and let an $[n_r,n_r-r,4]$ code  be quasi-perfect. Then a parity check  matrix $H_{r}$ of the code can be presented in the form
\begin{equation} \label{eq2_cor}
H_r =\left[ {\begin{array}{ccccccc}
B_{0,g}&|&B_{1,g}&|&&|&B_{D,g}\\
- - - &|&- - -&|&\ldots&|&---\\
H_{g+2}&|&H_{g+2}&|&&|&H_{g+2}
\end{array}} \right],
\end{equation}
where  $D=2^{r-g-2}-1$, $H_2=M$, $H_4=S$, $H_{g+2}$ is a parity check  matrix of a quasi-perfect  $[2^{g}+1,2^{g}+1-(g+2),4]$ code if $g\ge3$.
\end{corollary}
\begin{remark}\label{rem2}
By Corollary \ref{cor2_DT} a parity check matrix of any quasi-perfect binary code with length $2^{r-2}+2^{r-2-g}$ and redundancy $r$ can be created by $(r-g-2)$-fold applying of the doubling construction.
\end{remark}

As it is noted above, an arbitrary $[n,n-r,4]$ code is either a quasi-perfect code or shortening of some quasi-perfect code with $d=4$ and redundancy $r$. Therefore Theorem \ref{th_DT}, Corollaries \ref{cor1_DT}, \ref{cor2_DT}, and Remark \ref{rem2}, in fact, describe all binary linear codes with $d=4$ and length $\ge2^{r-2}+2$. It is why weight spectrum of codes obtained by the doubling construction \eqref{eq2_doubl} is an important problem.

The class of codes, say $\mathcal{D}$, obtained by the doubling construction is sufficiently wide. By \eqref{eq2_doubl}, the $[2^{r}-1,2^{r}-1-r,3]$ Hamming code and many its shortenings are included  in  $\mathcal{D}$.
 It follows  from Theorem \ref{th_DT} that $[2^{r-1},2^{r-1}-r,4]$ extended Hamming code and Panchenko code $\Pi _{r} $ (see below) belong to $\mathcal{D}$. Other numerous non-equivalent codes of  $\mathcal{D}$ can be obtained
by multiple application of the doubling construction to distinct quasi-perfect  $[2^{g}+1,2^{g}+1-(g+2),4]$ codes $\mathcal{C}_0$ with $g\in\{0,2,3,4,5,\ldots,r-3\}$, see \eqref{eq2_cor}. Examples of codes $\mathcal{C}_0$ can be
 found in \cite{DavydovTombakPPI,BrWeh,Weh} in algebraic and in geometrical forms. For instance, we give a parity check  matrix of a quasi-perfect  $[9,9-5,4]$ code.
\begin{equation*}
\left[\begin{array}{ccc}00000&|&1111\\10001 &|&0000 \\ 01001 & | &1001 \\
00101 & | & 0101 \\ 00011 &|& 0011 \end{array}\right].
\end{equation*}

The quasi-perfect codes $\Pi _{r} $ were proposed by V.I.\,Pan\-chenko in
paper \cite{Panchenko}. The $[n,n-r,4]$ code $\Pi _{r} $ has
length $n=5\cdot 2^{r-4} $, redundancy $r\ge
5$, and code distance $d=4$. (In paper \cite{DavydovTombakIEEE} the code $\Pi _{r} $ is denoted as
$\Pi$.)

The parity check $r\times 5\cdot 2^{r-4}$ matrix $P_{r} $ of Panchenko code $\Pi _{r} $ is the matrix $H_r$ of \eqref{eq2_cor} with $g=2$,
$D=2^{r-4} -1$, and $H_{g+2}=S$. So,
\begin{equation} \label{GrindEQ__2_7_}
P_{r} =\left[\begin{array}{ccccc} {B_{0,2} } & {B_{1,2} } & {B_{2,2} } & {\ldots } & {B_{D,2} } \\ {S} & {S} & {S} & {\ldots } & {S} \end{array}\right].
\end{equation}

Remind the known  \cite{DavydovTombakIEEE,Panchenko,AfDavZig2016} and important properties of Panchenko code and its shortenings:

 $\bullet$ For all $r$ and  $n$, there exist a shortened Panchenko code in which
the number of weight 4 codewords is close to the theoretical lower bound.

 $\bullet$ Independently of shortening algorithm, for all $r$ and  $n$,
the number of weight 4 codewords in Panchenko code and its shortenings is smaller than in a shortened extended Hamming code.

 $\bullet$ For $r=7$, $n\in\{32,33,\ldots,40\}$,  and $r=8$, $n\in\{72,73$, $\ldots,80\}$,
 Panchenko code and its shortenings by a special algorithm have the minimal number of weight 4 codewords among all other codes of the same length and redundancy.

 As the consequence of this property, Panchenko code has a small (often the minimal) probability of undetected error since this probability is essentially defined by the number of weight 4 codewords. In particular, it is important for error correction in computer memory \cite{AfDavZig2016,DavydovTombakIEEE}.

\section{Weight spectrum of  codes created by the doubling construction}

By Section \ref{sec_doubconstr}, a parity check matrix of any quasi-perfect binary code with $d=4$ can be created by multiple application of the doubling construction. Therefore, Theorems \ref{th_spect} and \ref{th_doubspectr} allow us to obtain weight spectrum of such code (and its dual) starting from weight spectrum of a short code.

We use notations introduced in the previous section. Also, for a code with redundancy $r$ we denote by $A_{w}^{(r)}$ the number of codewords of weight $w$ and by $A_{w}^{(r)\bot}$ the number of codewords of weight $w$ in the dual code.
\begin{theorem}\label{th_spect}
Let $d_r\le4$. Assume that an $[n_r,n_r-r,d_r]$ code $\mathcal{C}_r$ is created from a $[\frac{1}{2}n_{r},\frac{1}{2}n_{r}-r+1,d_{r-1}]$ code $\mathcal{C}_{r-1}$ by the doubling construction \eqref{eq2_doubl}. Then weight spectrum $\{A_{w}^{(r)},\, d_r\le w\le n_r\}$ of\/ $\mathcal{C}_r$ can be obtained from weight spectrum $\{A_{w}^{(r-1)},\, d_{r-1}\le w\le \frac{1}{2}n_r\}$ of $\mathcal{C}_{r-1}$ as follows:
\begin{equation} \label{eq3_1}
A_{2v}^{(r)}=\Delta_v^{(r)}+\sum\limits_{j=0}^{v-2}2^{2v - 2j - 1}A_{2v - 2j}^{(r-1)}\binom{\frac{1}{2}n_r-2v+2j}{j}
\end{equation}
where
\begin{equation*}
\Delta_v^{(r)}=\left\{\begin{array}{cccc}0&\text{if}&v&\text{odd}\\\binom{\frac{1}{2}n_r}{v}&\text{if}&v&\text{even}\end{array}\right.;
\end{equation*}
\begin{equation} \label{eq3_2}
A_{2v+1}^{(r)}=\sum\limits_{j=0}^{v-2}2^{2v - 2j}A_{2v+1 - 2j}^{(r-1)}\binom{\frac{1}{2}n_r-2v-1+2j}{j}.
\end{equation}
\end{theorem}
\begin{proof}
We consider structures of  weight $w$ codewords and the structures of the corresponding sets of  $w$ columns of a parity check matrix.

Let $u\in\{r,r-1\}$. Let $c_{w,u}$ be a weight $w$ codeword of the code $\mathcal{C}_u$. Denote by $H_u(c_{w,u})$ the set of $w$ columns of the matrix $H_u$ corresponding to the codeword $c_{w,u}$. By definition, the sum of all columns of $H_u(c_{w,u})$ is equal to zero.

We describe column sets of $H_r$ in \eqref{eq2_doubl}  with the help of column sets of $H_{r-1}$ placed in the left and right  sides of \eqref{eq2_doubl}.

\textbf{(i)} Let us consider all possible structures of codewords $c_{2v,r}$ of \textbf{even weight} $2v$ and the corresponding column sets $H_r(c_{2v,r})$  in the matrix $H_r$ of \eqref{eq2_doubl}. Every such column set consists of the following components:

$\bullet$ A column set $H_{r-1}(c_{2v - 2j,r-1})$ partitioned by two parts that are placed in the left and right sides of $H_r$.

$\bullet$ Two sets of the same $j$ columns of $H_{r-1}$ placed in the left and right sides of $H_r$. (These column sets are not connected with any codewords of $\mathcal{C}_{r-1}$.)

For $j=0,1,\ldots,v-2$ and for every codeword $c_{2v - 2j,r-1}$ of even weight, we explain summands of the formula \eqref{eq3_1}.

-- \emph{The summand $\sum\limits_{j=0}^{v-2}2^{2v - 2j - 1}A_{2v - 2j}^{(r-1)}\binom{\frac{1}{2}n_r-2v+2j}{j}$ of \eqref{eq3_1}.}\\
A column set $\Gamma=H_{r-1}(c_{2v - 2j,r-1})$  is partitioned by two parts. Every part contains an odd (resp. even)  number of columns if $j$ is odd (resp. even). The partition is executed by all possible ways. The number of the partitions is equal to $2^{2v - 2j - 1}$.
The obtained parts are placed in the left and right sides of $H_r$.

Also, in every of two submatrices $H_{r-1}$ of \eqref{eq2_doubl} we take the same set of $j$ columns
that do not belong  to $\Gamma$. The number of such $j$-sets is equal to $\binom{\frac{1}{2}n_r-2v+2j}{j}$. As a result, in the right side of $H_r$ we always take an even number of columns.

-- \emph{The summand $\Delta_v^{(r)}=\binom{\frac{1}{2}n_r}{v}$ of \eqref{eq3_1}.}\\
If $v$ is even then in every of two submatrices $H_{r-1}$ of \eqref{eq2_doubl} we take the same set of $v$ columns. The number of variants is equal to $\binom{\frac{1}{2}n_r}{v}$.

\textbf{(ii)} Let us consider all possible structures of codewords $c_{2v+1,r}$ of \textbf{odd weight} $2v+1$ and the corresponding column sets $H_r(c_{2v+1,r})$  in the matrix $H_r$ of \eqref{eq2_doubl}. Every such column set consists of the following components:

$\bullet$ A column set $H_{r-1}(c_{2v+1 - 2j,r-1})$ partitioned by two parts that are placed in the left and right sides of $H_r$.

$\bullet$ Two sets of the same $j$ columns of $H_{r-1}$ placed in the left and right sides of $H_r$. (These column sets are not connected with any codewords of $\mathcal{C}_{r-1}$.)

For $j=0,1,\ldots,v-2$ and for every codeword $c_{2v+1 - 2j,r-1}$ of odd weight, we explain the formula \eqref{eq3_2}.

A column set $\Gamma=H_{r-1}(c_{2v+1 - 2j,r-1})$  is partitioned by two parts. One  part, say $A_{odd}$, contains an odd number of columns, another part, say $B_{even}$, contains an even number of columns. The partition is executed by all possible ways. The number of the partitions is equal to $2^{2v - 2j}$.

If $j$ is odd  then the part $B_{even}$ (resp. $A_{odd}$) is placed in the left (resp. right) side of $H_r$.

If $j$ is even or $j=0$ then the part $A_{odd}$ (resp. $B_{even}$) is placed in the left (resp. right) side of $H_r$.

Also, in every of two submatrices $H_{r-1}$ of \eqref{eq2_doubl} we take the same set of $j$ columns
that do no belong  to $\Gamma$. The number of such $j$-sets is equal to $\binom{\frac{1}{2}n_r-2v-1+2j}{j}$. As a result, in the right side of $H_r$ we always take an even number of columns.
\end{proof}

Now we give the weight spectrum for codes dual to quasi-perfect ones.
\begin{theorem}\label{th_doubspectr}
Let $d_r\le4$. Assume that an $[n_r,n_r-r,d_r]$ code $\mathcal{C}_r$ is created from a $[\frac{1}{2}n_{r},\frac{1}{2}n_{r}-r+1,d_{r-1}]$ code $\mathcal{C}_{r-1}$ by the doubling construction \eqref{eq2_doubl}. Let $\frac{1}{2}n_{r}$ be even. Then weight spectrum $\{A_{w}^{(r)\bot},\, w\le n_r\}$ of the  $[n_r,r,d_r^\bot]$ code dual to $\mathcal{C}_r$ can be obtained from weight spectrum $\{A_{w}^{(r-1)\bot},\, w\le \frac{1}{2}n_r\}$ of the $[\frac{1}{2}n_{r},r-1,d_{r-1}^\perp]$ code dual to $\mathcal{C}_{r-1}$ as follows:
\begin{equation} \label{eq3_3}
A_{2v}^{(r)\bot} = A_v^{(r - 1)\bot}  + \left\{ {\begin{array}{ccc}
0&\text{if}&{2v \ne\frac{1}{2}n_r}\\
{{2^{r - 1}}}&\text{if}&{2v =\frac{1}{2}n_r}
\end{array}} \right..
\end{equation}
\end{theorem}
\begin{proof}
    We consider matrix \eqref{eq2_doubl} as a generator matrix of the dual code. If codeword of the dual code is created without  inclusion the top row, then its weight is equal to the doubled weight of the corresponding word formed from rows of matrix $H_{r-1}$. If the top row is included into codeword, its weight is equal to $\frac{1}{2}n_r$.
\end{proof}

\section{On correction of erasure patterns of high weight}
Knowledge of the weight spectrum of a code opens a way for calculation of very important probabilities for the code, like conditional probability of correct decoding of erasure patterns, probability of undetected error and so on. In binary codes, the number of parity check bits is  larger than code distance. That is a good reason to investigate a total ability of  codes to correct erasure patterns of high weights (equal to or greater than code distance).

The necessary condition for correction of weight $\rho$ erasure patterns  is the full rank of submatrix, consisting of columns of a code parity check matrix, corresponding to erased positions.

Let $S_{\rho}$ be the number of erasure patterns of weight  $\rho$, which can be corrected by a code (equivalently, for a code parity check matrix, $S_{\rho}$ is the number of distinct sets of $\rho$ linear independent columns  or the number of distinct $r\times\rho$ submatrices of the full rank).

For a code of length $n$, let $\delta_\rho=\frac{S_{\rho }}{\binom{n}{\rho}}$ be the conditional probability of correct decoding of erasure patterns of weight~$\rho$.

In further, for $[n,n-r,d]$ code with weight spectrum $A_0,A_1,\ldots,A_n$ we introduce the function
\begin{gather}\label{GrindEQ__2_1_}
\Psi (n,d,\rho )=\binom{n}{\rho }-\sum _{w=d}^{\rho}A_{w}\binom{n-w}{\rho-w},\quad d\le \rho \le r.
\end{gather}
This function gives a lower estimate of $S_{\rho}$, see \cite{AfDavZig2016_4,Popov}.

\begin{theorem}\label{th2.1}
For an $[n,n-r,d]$ code, the conditional probability $\delta_\rho$ and the value $S_{\rho }$ satisfy the following lower estimates:
\begin{equation}\label{GrindEQ__2_2_}
\delta _{\rho } \ge \frac{\Psi (n,d,\rho )}{\binom{n}{\rho }},\quad S_{\rho } \ge \Psi (n,d,\rho ),\quad d\le \rho \le r.
\end{equation}
In particular, the following equalities
\begin{equation*} \label{GrindEQ__2_3_}
\delta _{\rho } =\frac{\Psi (n,d,\rho )}{\binom{n}{\rho }},\quad S_{\rho } =\Psi (n,d,\rho ),
\end{equation*}
hold under the condition
\begin{equation*} \label{GrindEQ__2_4_}
\rho \le d+ \frac{d-1}{2}\,.
\end{equation*}
\end{theorem}
The proof of Theorem \ref{th2.1} is based on the fact that the value $S_{\rho }$ is equal to the difference between the total number of sets of $\rho$ columns of a parity check matrix and the number of patterns of $\rho$ linear dependent columns.

Now we use the known binomial approximation of weight spectrum of a binary linear code \cite{Cheung,KrasLits,MWS}
 $A_{w} \approx 2^{-z} \binom{n}{w }$, $ r-1< z\le r,\; w\ge d,$ where  $z$ is a real value taking into account (in principle) correction terms in the mentioned approximations and the weight region $w\ge d$. We obtain the following
 approximation of the function $S_\rho$  for the region $d\le \rho\le r$.
\begin{align*}
&S_{\rho} \geq\binom{n}{\rho }-\sum _{w=d}^{\rho}A_{w}\binom{n-w}{\rho-w}\\
 &\approx \binom{n}{\rho }-2^{-z} \sum _{w=d}^{\rho}\binom{n}{w}\binom{n-w}{\rho-w}\\
  &=\binom{n}{\rho }-2^{-z} \binom{n}{\rho }\sum _{w=d}^{\rho}\binom{\rho}{w} .
\end{align*}
From here, using \cite[Lemma 10.8]{MWS}, we obtain an estimate of the conditional probability $\delta_\rho$ of correct decoding of erasure patterns of high weight~$\rho$.
\begin{align*}
&\delta _{\rho} \geq\frac{S_{\rho} }{\binom{n}{\rho }}\approx 1-2^{-z} \sum _{w=d}^{\rho}\binom{\rho}{w}\\
& \approx1-2^{-z}\cdot2^{\rho H(d/\rho)}\ge 1-2^{\rho-z} ,\; d\le \rho<z,
\end{align*}
where $H(d/\rho)$ is the binary entropy.

The proposed estimate shows that for a fixed $r$, the probability $\delta_\rho$ decreases exponentially with growth of $\rho$. Therefore the reasonable extended region of correctable erasure patterns is $\rho<2d$.

The following lemma allows us to improve estimates of Theorem \ref{th2.1} using a recursive approach.
\begin{lemma}
    Any set of $\rho$ linear dependent columns of a parity check matrix is an union of $w$ columns with the zero sum (corresponding to a weight $w$ codeword ) and a set of $\rho-w$ linear independent columns, where $d\le w\le\rho$.
\end{lemma}

We give a recursive form of function of type \eqref{GrindEQ__2_1_} :
\begin{equation*}\tilde{\Psi }(n,d,\rho )=\binom{n}{\rho}-\sum _{w=d}^{\rho }A_w(n) \tilde{\Psi }(n-w,d,\rho -w),
\end{equation*}
where $A_w(n)$ is the number of weight $w$ words in a (shortened) code of length $n$.

A recursive estimate of the conditional probability of correct decoding of erasure patterns of weight~$\rho$ and  the first and second steps of the recursion has the form, respectively,
\begin{align*}
&\tilde{\delta }\left(n,d,\rho \right)=\frac{\tilde{\Psi }(n,d,\rho )}{\binom{n}{\rho}}\\
& =1-\sum _{w=d}^{\rho }A_w(n) \tilde{\delta }(n-w,d,\rho -w)\frac{\binom{n-w}{\rho-w}}{\binom{n}{\rho}}; \\
&\tilde{\delta_2 }\left(n,d,\rho \right)
=1-\sum _{w_1=d}^{\rho }A_{w_1}(n) \frac{\binom{n-w_1}{\rho-w_{1}}}{\binom{n}{\rho}}\,\times\\
 &\times\left[1-\sum _{w_2=d}^{\rho -w_1}A_{w_2}\left(n-w_1\right)\frac{\binom{n-w_1-w_2}{\rho-w_1-w_2}}{\binom{n-w_1}{\rho-w_1}}  \right].
\end{align*}

\section{Application to memory }
An important area for application of quasi-perfect codes is computer memory (Flash or SSD). Their ability to  correct a big number of erasures instead of one error and very low probability of undetected error gives us a strong incentive to investigate the conditional probability of correct decoding for erasure patterns of high weight.
As an example, useful for application, we give two tables: the first one for conditional probability of correct decoding for erasure patterns of weights higher the code distance and the second one for the probability (unconditional) of decoding failure in memory channel with different error probability for the product of Panchenko codes.

Decoding algorithm for  product of Panchenko codes consists of following steps.
\begin{enumerate}
  \item Error detection in rows and columns of the received word (in parallel).
  \item Check (in parallel) of the detected row (column) list for correctability as erasure pattern.
  \item Correction of the chosen erasure pattern (row or\, column) and output.
\end{enumerate}

Check for correctability is executed in extended area up to $d^{+}$ erasures.

Table \ref{tab1} gives a comparison between Hamming and Pan\-chenko codes with 7 and 8 parity symbols. We can see from the table that extended decoding with  correction of  4, 5, 6, 7 erasures has decreasing probability from 1 up to 1/2 (approximately).

Table \ref{tab2} demonstrates fast decreasing of the probability of decoding failure for fixed number of parity bits with extension of the decoding area for product of two Panchenko codes. We can see from the second table  fast decreasing of the failure probability with extension of the decoding area from 3 up to 6 erasures.
\begin{table}[htbp]
\caption{Conditional probability $\delta_\rho$ of correct decoding of erasure patterns of
weight~$\rho$ for Hamming and Panchenko codes}
\label{tab1}
\begin{center}
\begin{tabular}{|c|c|c|c|c|c|}
  \hline
code & $r$ &  $\rho=d=4$ & $\rho =5$ & $\rho =6$ & $\rho =7$ \\\hline
Hamming & 7 & 0.9836 & 0.9180 & 0.7469 & 0.4121 \\
Panchenko & 7 & 0.9870 & 0.9287 & 0.7656 & 0.4306 \\\hline
Hamming   & 8& 0.9920 & 0.9600 & 0.8741 & 0.6879 \\
Panchenko & 8 & 0.9934 & 0.9647 & 0.8830 & 0.6996 \\
  \hline
\end{tabular}
\end{center}
\end{table}

\begin{table}[htbp]
\begin{center}
\caption{Failure probability for product of Panchenko codes $[72,64,4]$}
\label{tab2}
\begin{tabular}{|c|c|c|c|c|c|} \hline
$p$  & $10^{-1} $ & $10^{-2} $ & $5\cdot 10^{-3} $ & $10^{-3} $ & $5\cdot 10^{-4} $  \\ \hline
$d^{+} =3$ & 1 & 0,996 & 0,250 & 1,1e-09 & 2,3e-14 \\ \hline
$d^{+} =4$ & 1 & 0,988 & 0,092 & 1,6e-12 & 5,1e-18  \\ \hline
$d^{+} =5$ & 1 & 0,967 & 0,027 & 7,0e-14 & 1,045e-18  \\ \hline
$d^{+} =6$ & 1 & 0,926 & 0,008 & 5,8e-14 & 1,029e-18  \\ \hline
\end{tabular}
\end{center}
\end{table}

\section*{Acknowledgment}
The research  was carried out at the IITP RAS at the expense of the Russian
Foundation for Sciences (project 14-50-00150).

\end{document}